%% file: root.tex
\documentclass[conference]{IEEEtran}

\usepackage{amsmath,amssymb,amsthm}
\usepackage{epsfig,graphicx}
\usepackage[svgnames,rgb]{xcolor}
\usepackage[color=yellow]{pdfcomment}
\usepackage{colortbl,hhline}
\usepackage{algorithmic,algorithm,lipsum}

\usepackage[mathscr]{eucal}
\usepackage{amsbsy}
\usepackage{bm}
\usepackage{fixltx2e}
\MakeRobust{\overrightarrow}

\include{macrofile}

\title{Group-invariant Subspace Clustering}

\author{\IEEEauthorblockN{Shuchin Aeron}
\IEEEauthorblockA{Dept. of Electrical and Computer Engineering\\
Tufts University, Medford, MA 02155\\
Email: shuchin@ece.tufts.edu}
\and
\IEEEauthorblockN{Eric Kernfeld}
\IEEEauthorblockA{Dept. of Statistics\\
University of Washington, Seattle, WA\\
Email: eric.kern13@gmail.com}}

\newtheorem{theorem}{Theorem}[section]
\newtheorem{lemma}[theorem]{Lemma}

\newtheorem{example}[theorem]{Example}

\begin{document}
\maketitle

\begin{abstract}
In this paper we consider the problem of group-invariant subspace clustering where the data is assumed to come from a \emph{union of group-invariant subspaces} of a vector space, i.e. subspaces which are invariant with respect to action of a given group. Algebraically, such group-invariant subspaces are also referred to as submodules. Similar to the well known Sparse Subspace Clustering approach where the data is assumed to come from a union of subspaces, we analyze an algorithm which, following a recent work \cite{Kernfeld:2014wta}, we refer to as Sparse Sub-module Clustering (SSmC). The method is based on finding group-sparse self-representation of data points. In this paper we primarily derive general conditions under which such a group-invariant subspace identification is possible. In particular we extend the geometric analysis in \cite{Soltanolkotabi:2012ia} and in the process we identify a related problem in geometric functional analysis.
\end{abstract}

\section{Introduction}
In this paper we consider the problem of group-invariant unsupervised clustering of data points. To give some examples where such a scenario may arise, consider the problem of clustering images \cite{Basri:2003wd,Chen:2012eo,Hastie:1998wo}, where due to shifts in camera position or minor changes in the pose, the images can be arbitrarily shifted or in some cases rotated. In such cases either a pre-processing step is performed to align/center the images, after which one can employ unsupervised clustering approaches such as in \cite{Elhamifar:2012uz}. The problem that we address in this paper is whether such a pre-processing step can be eliminated given the group (essentially a set of transformations that commute with each other) with respect to which, we would like to keep the invariance in clustering the data. 
In this context, we note a related work \cite{Kernfeld:2014wta} in which, the authors propose a (horizontal) shift-invariant clustering of images using linear algebraic constructs developed in \cite{Kilmer:2013kx}. This work extends the analysis as well as the concept therein. For example, we don't assume that the submodules need to be \emph{free} and disjoint. In addition the proposed method works for any group. We also note that the geometric analysis presented in this paper can also be potentially applied and related to the performance bounds on group-sparse recovery problems, such as the block sparse recovery problem considered in \cite{Elhamifar:2012}. Nevertheless, our framework is group-theoretic and we derive performance bounds in terms of a novel notion of \emph{group-subspace incoherence}. 

The rest of the paper is organized as follows. In Section \ref{sec:prob_setup} we describe the problem set-up where the data is modeled as coming from a union of group-invariant subspaces or submodules. In Section \ref{sec:alg} we present the algorithm for performing sparse sub-module clustering and analyze it with respect to a semi-random model in Section \ref{sec:analysis}. In Section \ref{sec:sims} we present simulation results to show the effectiveness of the proposed approach. Finally we provide conclusions and future research directions in Section \ref{sec:conclude}.

\textbf{Notation}: In the following we will use capital boldface letters $\M{X}$ to denote matrices/2-D data points, lowercase boldface letters $\V{x}$ to denote the column or row vectors. For a matrix $\M{X}_i$ denotes the $i$-th column and $\M{X}_{\{i\}}$ denotes the $i$-th row. Occasionally we will need to form 3-D arrays and we denote them by uppercase calligraphic bold-face letters $\T{X}$. Additional notation is introduced as needed.


\section{Problem set-up}
\label{sec:prob_setup}

Denote by $\mb{R}^n$ the real vector space over $\mb{R}$. Given an abelian (commutative) group $\mb{G}$ of order $N_G$, which acts on $\mb{R}^n$ through its \emph{linear representation} \cite{Serre_book} $\M{L}_g, g \in \{1,2,...,N_G\}$, $\M{L}_g: \mb{R}^n \rightarrow \mb{R}^n$ such that $\M{L}_{g_2}\M{L}_{g_1}\V{v} = \M{L}_{g_1}\M{L}_{g_2}\V{v}, \forall \V{v} \in \mb{R}^n$. Then $\mb{R}^n$ is said to be a $\mb{G}$-module\footnote{Note $\M{L}_g \in GL_n(\mb{R}^{n})$, where $GL_n(V)$ denotes the general linear group, i.e. the group of invertible mappings on the vector space $V$}.



\textbf{Submodule}: A subspace $\mb{S}$ of $\mb{R}^n$ such that for all $\V{u} \in \mb{S}$, $\M{L}_g \V{u} \in \mb{S}$ is called as a submodule (w.r.t. $\mb{G}$). Note that a submodule is essentially a $\mb{G}$-invariant subspace of $\mb{R}^n$. 

Note that without loss of generality $\M{L}_1$ can be taken to be the identity element so that $\M{L}_1 = \M{I}_{n\times n}$, the $n \times n$ identity matrix.  

\begin{example}
Let us illustrate the algebraic set-up through an example. We consider the setting of \cite{Kernfeld:2014wta} where data points are considered as images of size $n_1 \times n_2$ (can be embedded in a vector space $\mb{R}^{n_1 n_2}$). The group is the cyclic group of shifts along the columns of the images and is of order $N_G = n_2$. The action of this group on $\mb{R}^{n_1 n_2}$ can be captured by $N_G$ matrices obtained by taking the Kronecker product of identity matrix of size $n_1 \times n_1$ with the $n_2$ matrices for cyclically permuting the columns. It can be seen that the generative model for the submodules used in \cite{Kernfeld:2014wta} is precisely the one as outlined above.
\end{example}

\begin{example}
For the previous example, one can consider a discretized rotation group rotating the images around the center. In particular one can take the direct product of the shift and rotation groups\footnote{Note that direct product of Abelian groups is Abelian, but direct product of cyclic groups is not cyclic}, increasing the order of the group to capture shift and rotational invariance.
\end{example}

We now present the formal problem statement.

\subsection{Problem statement: The union of submodules model} 

We are given a set of data points collected as columns of a matrix $\M{X}$, i.e. $\M{X}_{i}, i = 1,2,...,N$ such that, $\M{X}_i \in \bigcup_{\ell =1}^{L} \mb{S}_{\mb{G}}^{(\ell)}$, where $\mb{S}_{\mb{G}}^{(\ell)}$ is a $\mb{G}$-invariant subspace of dimension $d_\ell$, i.e. submodule of dimension $d_\ell$, for $\ell = 1,2,...,L$. The problem is to identify $\mb{S}_{\mb{G}}^{(\ell)}$ and cluster the points such that within each cluster the points belong to the same submodule.

We now present an algorithm for solving this problem. The main approach is very similar to that of sparse subspace clustering albeit \emph{\textbf{we make use of the group structure in identifying the submodules}}.

\subsection{The Algorithm: Sparse Submodule Clustering (SSmC)}
\label{sec:alg}

For the submodule clustering we propose the following algorithm outlined below. 

\begin{enumerate}
\item For each $i$ solve the following convex optimization problem, which is essentially computing group-subspace affinity in self-representation. 
\begin{align}
&\min \| \rs(\V{c}_i)\|_{1,2} \nonumber \\
& \mbox{s.t.}\,\, \M{X}_{-i,\mb{G}}\V{c}_i = \M{X}_i \,\, 
\end{align}
where
\begin{itemize}
\item $\| \cdot \|_{1,2}$ denotes the group-sparse norm equal to the $\ell_1$ norm of the vector of $\ell_2$ norms of the \emph{rows}.
\item $\M{X}_i$ denotes the $i$-th column of $\M{X}$.
\item $\M{X}_{-i}$ denotes the matrix with the $i$-th column removed. 
\item $\M{X}_{-i,\mb{G}}$ is an $n \times (N-1)N_G$ matrix $=  \left[ \M{L}_1 \M{X}_{-i}, \M{L}_2 \M{X}_{-i}, ..., \M{L}_{N_G} \M{X}_{-i} \right]$
\item For \emph{any} vector $\V{c}$ (row or column) of size $K N_G$, $\rs(\V{v})$ is a matrix of size $K \times N_G$, where the first column corresponds to the first $K$ elements of $\V{v}$, the second column corresponds to the next $K$ elements and so on. 
\end{itemize}

\item Form an $N \times N$ affinity matrix $\M{C}$, where the $i$-th column $\M{C}_{i}$ is the vector consisting of the $\ell_2$-norms of the rows of $\rs(\V{c}_i)$. 

\item Let $\M{W} = |\M{C}| + |\M{C}|^\top$ and perform spectral clustering \cite{njwspectral,Luxburg:2007dq}  using $\M{W}$ as the affinity or the weight matrix. 

\item Output the clusters. 
\end{enumerate}

\section{Analysis of the algorithm}
\label{sec:analysis}

In order to derive meaningful performance guarantees for the proposed algorithm, we consider the following generative model. Let $\M{X}^\ell$ denote the set of vectors in $\M{X}$ which belong to submodule $\ell$. Let $$ \M{X}^\ell = \M{Q}_{\mb{G}}^{(\ell)} \M{A}^{\ell} ,$$ where the $N_\ell$ columns of $\M{A}^\ell$ are drawn from the unit sphere $\mc{S}^{d_\ell-1}$ and $\M{Q}_{\mb{G}}^{(\ell)}$ is a matrix with orthonormal columns, whose columns span a $\mb{G}$-invariant subspace. 

Our analysis will rest on considering the following primal and dual optimization problems and their solutions. 
\begin{align}
\mc{P}(\V{y},\M{M}): \,\, &\min \| \rs(\V{c})\|_{1,2} \\
& \mbox{s.t.}\,\, \M{M}\V{c} = \V{y}
\end{align}
where $\M{M}$ is a matrix of size $n \times N N_G$. The dual problem can be written as, 
\begin{align}
\mc{D}(\V{Y},\M{M}): \,\, &\max \,\, \la \V{Y}, \V{\lambda}\ra \\
& \mbox{s.t.}\,\, \|\rs(\M{M}^{\top}\V{\lambda})\|_{\infty,2} \leq 1
\end{align}
Here $\| \M{X} \|_{\infty,2}$ norm of a matrix $\M{X}$ is the $\ell_\infty$ norm of the vector of $\ell_2$ norm of the rows. 

\textbf{Notation}: In the following we use $\M{C}_{\{i\}}$ to denote the $i$-th row of the matrix $\M{C}$. For an $n \times N N_G$ matrix $\M{M}$, given a set $S \subset \{1,2,...,N\}$, the matrix $\M{M}_S$ is a $n \times |S|\cdot N_G$ matrix, formed as follows: First reshape $\M{M}$ to a 3-D array $\T{M} \in \mb{R}^{n \times N \times N_G}$.  Form a $n \times |S| \times N_G$ 3-D array $\T{M}_S$ by taking the lateral slices $\T{M}(:,i,:), i \in S$. Then reshape the resulting 3-D array $\T{M}_S$ back to a $n \times |S|N_G$ matrix $\M{M}_S$. 

We have the following Lemma. 
\begin{lemma}
\label{lem:prim_dual}
If there exists,
\begin{enumerate}
\item A $\V{C}$ satisfying $\V{Y} = \M{M}\V{C}$ such that the row-support (the number of non-zero rows) $S$ of $\M{C}= \rs(\V{c})$ satisfies $S \subseteq T$, and
\item A dual certificate $\V{\lambda}$ satisfying 
\begin{align}
\label{eq:c1} [\rs(\M{M}_{S}^{\top} \M{\lambda})]_{\{i\}}  = \frac{\M{C}_{\{i\}}}{\|\M{C}_{\{i\}}\|_2} \\
\label{eq:c2}  \|\rs(\M{M}_{T\cap S^c}^{\top} \V{\lambda})\|_{\infty,2} \leq 1 \\
\label{eq:c3} \|\rs(\M{M}_{T^c}^{\top}\V{\lambda})\|_{\infty,2} < 1
\end{align}
\end{enumerate}
then all optimal solutions $\V{Z}$ to $\mc{P}(\V{Y},\M{M})$ obey that the rows of the matrix $\M{Z} = \rs(\V{z})$ corresponding to set $T^c$ are zero. 
\end{lemma}

\begin{proof}
The proof follows by proceeding through the same arguments as that of Lemma 7.1 in \cite{Soltanolkotabi:2012ia} and substituting the ${\tt sgn}$ function with $$[{\tt sgn}(\M{C})]_{\{i\}} = \frac{\M{C}_{\{i\}}}{\|\M{C}_{\{i\}}\|_2}.$$
\end{proof}

Based on this Lemma we now derive a sufficient condition for submodule identification.

\subsection{A geometric condition for submodule identification}

In the following we will assume that the optimization problems have a unique solution. This assumption is to avoid unnecessary technicalities and does not affect the core analysis and the main results.

In the following we will assume that one is given $N_\ell$ points for submodule $\mb{S}_\ell$. Closely mirroring the analysis in \cite{Soltanolkotabi:2012ia} we proceed in steps as follows.

\begin{itemize}
\item Let $\V{c}_{i}^{\ell}$ be the (unique) solution to the primal program $\mc{P}({\V{A}}_{i}^{\ell}, {\M{A}}_{-i,\mb{G}}^{\ell})$ where ${\V{A}}_{i}^{\ell}$ denotes the $i$-th column of ${\M{A}}^{\ell}$ and,
\begin{align} 
{\M{A}}_{-i,\mb{G}}^{\ell} = \M{Q}_{\mb{G}}^{\ell \top}[ \M{L}_1 \M{X}_{-i}^{\ell}, \M{L}_2 \M{X}_{-i}^{\ell},..., \M{L}_{N_G} \M{X}_{-i}^{\ell}]
\end{align}

\item Let $\V{\lambda}_i^\ell$ be the solution to the dual problem, $\mc{D}({\V{A}}_{i}^{\ell}, {\M{A}}_{-i,\mb{G}}^{\ell})$.
\item Define the \emph{group-dual direction} for points $\V{a}_{i}^{\ell}$, $i = 1,2,...,N_\ell$ via 
\begin{align}
{\V{v}}_i^\ell = \M{Q}_{\mb{G}}^{\ell} \frac{\V{\lambda}_i^\ell}{\|\V{\lambda}_i^\ell \|_2}
\end{align}
\end{itemize}

Let $\V{\nu}_i^\ell = \M{Q}_{\mb{G}}^{\ell} \V{\lambda}_i^\ell$.  Then it is easy to see that 
\begin{align*}
&\rs(\V{c}) \nonumber \\
& = [ 0_{N_1 \times N_G}; 0_{N_2 \times N_G};...; \rs(\V{c}_i^\ell); ...; 0_{N_{K} \times N_G} ]\, ,
\end{align*} 
and $\V{\nu}_{i}^{\ell}$ satisfy conditions of Equations (\ref{eq:c1}) and (\ref{eq:c2}) of Lemma \ref{lem:prim_dual}. In order to satisfy the inequality (\ref{eq:c3}) of the Lemma \ref{lem:prim_dual} one needs to satisfy,
\begin{align}
\| \rs(\M{X}_{\mb{G}}^{(k)\top} {\V{\nu}}_i^\ell)\|_{\infty,2} < 1
\end{align}
for all $(k) \neq \ell$ and for all $i= 1, 2, ..., N_\ell$.  Equivalently these conditions translate to the condition, 
\begin{align}
\left\| \M{X}_{j,\mb{G}}^{(k) \top} \V{v}_i^\ell\right\|_{2} \|\V{\lambda}_i^\ell \|_2< 1
\end{align}
for all $j = 1,2,...,N_k$. Now we need to bound $\|\V{\lambda}_i^\ell \|_2$. Note that $\V{\lambda}_i^\ell$ lies in the set $P$ where 
\begin{align} 
P = \{ \V{\lambda}_{i}^{\ell}: \|\M{A}_{-i, \mb{G}}^{\ell \top} \lambda_i^\ell\|_{\infty,2} \leq 1\}.
\end{align} 
The circumscribing radius $\mc{R}(P)$ of this centro-symmetric set is a bound on $\|\V{\lambda}_i^\ell \|_2$. Note that the polar $P^\circ$ of this set is given by, 
\begin{align} 
P^\circ = \{ \V{z}: \V{z} = \M{A}_{-i, \mb{G}}^{\ell} \V{b},\,\, : \|\rs(\V{b})\|_{1,2} \leq 1\}
\end{align} 
Using the polar duality \cite{Rockafellar} and the inverse relation between the in-radius $r(P^\circ)$ and $\mc{R}(P)$ \cite{Soltanolkotabi:2012ia} we obtain, 
\begin{align} 
\|\V{\lambda}_i^\ell \|_2 \leq \mc{R}(P) = \frac{1}{r(P^\circ)}
\end{align} 
Therefore a sufficient condition for group-invariant subspace identification becomes, 
\begin{align} 
\label{eq:cond1}
\left\|\M{X}_{j,\mb{G}}^{(k) \top} \V{v}_i^\ell \right\|_{2} \leq r(P^\circ)
\end{align}
Note that $\V{\xi} = \M{X}_{j,\mb{G}}^{(k) \top} \V{v}_i^\ell$ is a vector of size $N_G$ with the $g$-th element $\V{\xi}_g$ given by $$\V{\xi}_g = \V{a}_j^\top \M{Q}_{\mb{G}}^{(k)\top} \M{L}_g^\top  \M{Q}_{\mb{G}}^{(\ell)}\frac{\V{\lambda}_i^\ell}{\|\V{\lambda}_i^\ell \|_2}.$$

Using this analysis, our main result is based on the following notion of \emph{affinity between submodules} --- Given $\M{L}_g, \forall g \in \mb{G}$, and $\M{Q}_{\mb{G}}^{\ell}$ and $\M{Q}_{\mb{G}}^{k}$, the \emph{submodule-affinity} between two $\mb{G}$-submodules is defined as,
\begin{align}
\sqrt{\sum_{g}\|\M{Q}_{\mb{G}}^{(k)\top} \M{L}_{g}^{\top} \M{Q}_{\mb{G}}^{(\ell)} \|_{F}^{2}}.
\end{align}
Note that this measure of affinity is measuring the total affinity between the submodules under the group action and in general is \emph{larger} than the affinity $ \|\M{Q}_{\mb{G}}^{(k)\top} \M{Q}_{\mb{G}}^{(\ell)} \|_{F}$ between the submodules treated as subspaces. 

Then one may wonder \textbf{why using the SSmC algorithm has any benefit over using the SSC algorithm?} We will attempt to address this question after deriving the identifiability conditions for a \emph{semi-random} generative model in the next section.

\section{Results for the semi-random model}

In the semi-random generative model the $N_\ell$ columns of $\M{A}^\ell$ are drawn uniformly randomly from the unit sphere $\mc{S}^{d_\ell-1}$.  For the semi-random model, we have the following result. 

\begin{figure*}
\begin{align}
\label{eq:main2}
P \left\{ \| \rs(\M{X}_{\mb{G}}^{(k)\top} {\V{v}}_i^\ell)\|_{\infty,2} \leq C_1(t,\Delta) \frac{\sqrt{\sum_{g}\|\M{Q}_{\mb{G}}^{(k)\top} \M{L}_{g}^{\top} \M{Q}_{\mb{G}}^{(\ell)} \|_{F}^{2}}}{\sqrt{d_\ell} \sqrt{d_k}}  \right\} \geq 1 - C_2(t,\Delta) e^{-2t}
\end{align}
\end{figure*}

\begin{theorem}
\label{thm:main}
Under the semi-random model for any two submodules $k$ and $\ell$, given $\Delta, t >0$, Equation (\ref{eq:main2}) holds with $$ C_2(t,\Delta) = \dfrac{4 N_G}{(N_k +1) \Delta^2 } e^{-2t}$$ and $$ C_1(t, \Delta) = 4 (\log(N_k+1) + \log \Delta + t).$$ 
\end{theorem}
\begin{proof}
The proof follows from Lemma \ref{lem:1} below and using Lemma 7.5 in \cite{Soltanolkotabi:2012ia}.
\end{proof}

\begin{lemma} 
\label{lem:1}
Under the semi-random model, $\frac{\V{\lambda}_i^\ell}{\| \V{\lambda}_i^\ell\|_2}$ are uniformly distributed on the unit sphere $\mc{S}^{d_\ell-1}$.
\end{lemma}
\begin{proof}
The result follows along the same lines as in the Proof of Step 2, section 7.2.2. in \cite{Soltanolkotabi:2012ia}.
\end{proof}

Therefore, for a suitable choice of $\Delta$ if the condition,
\begin{align}
C_1(t, \Delta) \frac{\sqrt{\sum_{g}\|\M{Q}_{\mb{G}}^{(k)\top} \M{L}_{g}^{\top} \M{Q}_{\mb{G}}^{(\ell)} \|_{F}^{2}}}{\sqrt{d_\ell} \sqrt{d_k}} \leq r(P^{\circ})\,\, ,
\end{align}
is satisfied then then the SSmC algorithm correctly identifies (pair-wise) the subspaces. Using a union bounding argument over all the submodules one can see that under sub-module incoherence SSmC algorithm clusters the points correctly with high probability.

\subsection{Submodule Vs Subspace clustering}

Let us now compare the condition derived above for submodule clustering with the case when one \emph{doesn't exploit any knowledge of the group with respect to which, the subspaces are invariant}. Using the result in \cite{Soltanolkotabi:2012ia} we note the following condition for correct (pair-wise) subspace identification, 
\begin{align}
C_1(t, \Delta) \frac{\sqrt{\|\M{Q}_{\mb{G}}^{(k)\top} \M{Q}_{\mb{G}}^{(\ell)} \|_{F}^{2}}}{\sqrt{d_\ell} \sqrt{d_k}} \leq r(P_{ssc}^{\circ})\,\, ,
\end{align}
where $$P_{ssc}^{\circ} = \{\tilde{\V{z}}: \tilde{\V{z}} =  \M{A}_{-i}^{\ell \top} \tilde{\V{b}},\,\, \| \tilde{\V{b}} \|_{1} \leq 1\}.$$
Now whether using the sparse sub-module clustering algorithm is useful, depends on whether how large the in-radius $r(P^{\circ})$ is, compared to $r(P_{ssc}^{\circ})$ and this in turn depends on the group as well as the number of points $N_\ell$ per subspace. 

With $\Delta = N_\ell L$ and if $N_\ell \approx d_\ell$ then following the results in \cite{Soltanolkotabi:2012ia}, we have the following lower bound on the in-radius of the set $r(P_{ssc}^{\circ})$,
\begin{align}
r(P_{ssc}^{\circ}) \geq c_0 \sqrt{\frac{\log \frac{N_\ell}{d_\ell}}{d_\ell}}\, .
\end{align}
for some fixed positive constant $c_0$. So compared to this lower bound, for SSmC to be beneficial over SSC, the in-radius $r(P^\circ)$ must scale with $N_G$. This fact may appear intuitive as we are adding volume to the sysmmteric body by adding the \emph{spherical caps} to the polytope, see Figure \ref{fig:geom}. We therefore have the following \textbf{conjecture} for a lower bound on the in-radius of $P^\circ$.

\begin{lemma} -- [\textbf{Conjecture}]
If $\M{L}_g$ are unitary and if for any random vector $\V{a}$ the matrix $[ \M{L}_1 \V{a}, \M{L}_2 \V{a},...,\M{L}_{N_G} \V{a} ]$ is full rank and if $d_\ell = \beta_\ell N_G, \beta_\ell > 1$ for all $\ell \in \{1,2,...,L\}$, then 
\begin{align}
r(P^{\circ}) \geq c_0 \sqrt{N_G} \sqrt{\frac{\log \frac{N_\ell}{d_\ell}}{d_\ell}}
\end{align}
\end{lemma}
At this juncture we are not able to prove or disprove this result by using existing results known for random polytopes \cite{Dafnis:2009hf,AlonsoGutierrez:2008tm,Pivovarov:2010fd}. The set $P^\circ$ is not a polytope but is formed out of a polytope with added smooth caps, see Figure \ref{fig:geom}. We leave this as an open problem to be solved in future. 

If this conjecture is true, then the benefits of SSmC are immediately clear. This is because the submodule incoherence compares the incoherence of the subspaces under all group actions and overall this averaged criteria helps distinguishability.

In the next section we will show the effectiveness of the proposed algorithm on some real data sets. The superior performance of the SSmC algorithm on one of these sets served as the primary motivation for the analysis carried out in this paper.  

\begin{figure}[htbp]	
\centering
 	\includegraphics[height= 1.5in, width = 3 in]{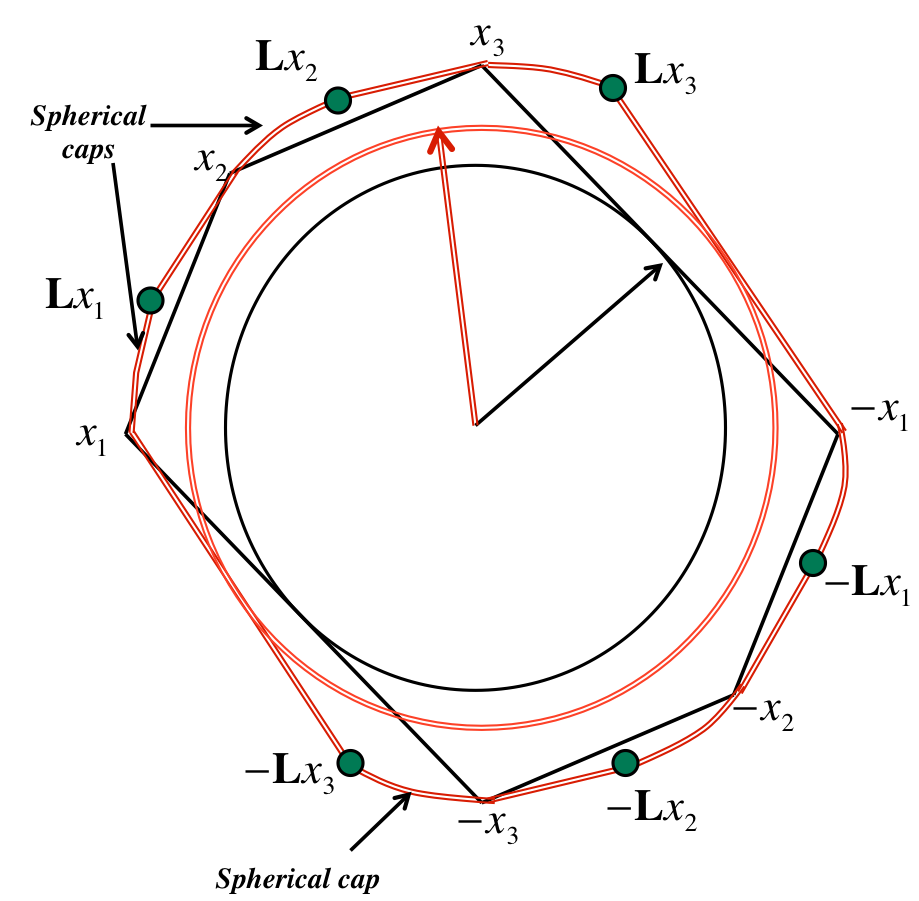}
 	\caption{Geometry of submodule clustering. $N_G$ more extreme points are added with spherical caps connecting the group of points generated by a data point. The conjecture states that due to the addition of spherical caps, the in-radius increases. This example figure depicts the increase in in-radius as a result of adding points for the set $P^\circ$ (shown in double red-lines) compared to the in-radius of the set $P_{ssc}^{\circ}$ (shown in single solid black lines).}
	\label{fig:geom}
	\end{figure} 

 

\section{Simulations}
For simulations on real data we solve the following,
\begin{align}
&\min_{\V{c}_i} \| \rs(\V{c}_i)\|_{1,2} + \lambda \| \M{X}_{-i,\mb{G}}\V{c}_i - \M{X}_i\|_{F}^{2}  
\end{align}
where $\lambda > 0 $ is the regularization parameter that penalizes the model mismatch. In the following we exhaustively search for the optimal $\lambda$ that yields the best clustering performance.

\label{sec:sims} 
\subsection{Yale database}
We now compare the performance of the SSmC algorithm from \cite{Kernfeld:2014wta} on the cropped Yale Face database -- http://vision.ucsd.edu/~leekc/ExtYaleDatabase/ExtYaleB.html, see Figure \ref{fig:YaleBCrop} for example images. For the experiments the down-sampled data base is directly taken from \cite{Elhamifar:2012uz} available through the authors' website. 

\textbf{Note}: that in \cite{Elhamifar:2012uz}, the authors use a outlier rejection version of SSC for clustering the faces. While one can also add a provision for sparse outlier rejection in SSmC we will not use it here. The main reason is that it will require us to select and optimize over two regularization parameters, which makes the problem of searching for the optimal parameters computationally challenging.

In Table \ref{tab:yale} we present the error performance for clustering $5$ subjects for cases when using $20,25,30$ images per person that are randomly chosen from the $64$ images for each. Note that the clustering performance of SSC is slightly better than SSmC but the difference in performance is not much. In this case since the data is extremely well represented by union of subspaces, the results agree well with the observation that SSC is a special case of SSmC in this case when the data is very close to union of subspaces there is perhaps no gain in using the SSmC approach.

\begin{figure*}	
\centering
 	\includegraphics[height= 1.5in, width = 4 in]{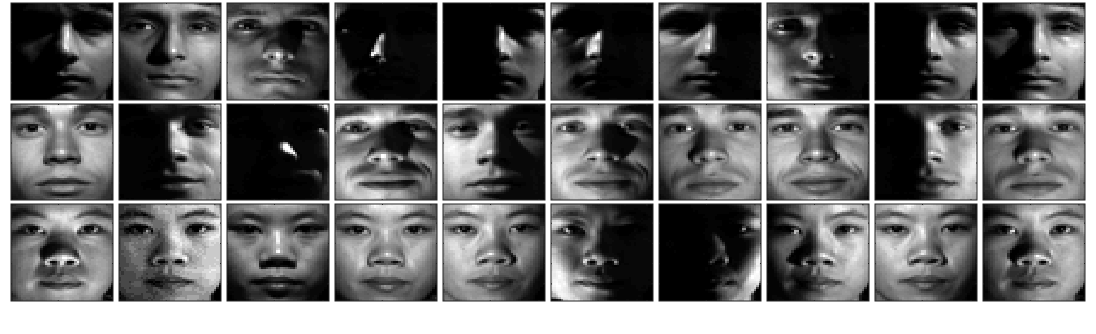}
 	\caption{Examples of faces from 3 subjects from the Cropped YaleB data set taken under different lighting conditions. Note that centering and alignment is done on the images to focus mainly on the face and there is no pose variation.}
	\label{fig:YaleBCrop}
	\end{figure*}   

\begin{table*}
\begin{center}
\begin{tabular}{ |c|c| c|c| }
\hline
 SSC & $20$ images per sub. & $25$ images per sub. & $30$ images per sub. \\
 \hline 
 SSmC & 18.60 \% & 18.56 \% & 14.93 \% \\  
 \hline
 SSC & 17.40 \% & 18.40 \% & 14.53 \% \\
 \hline
\end{tabular}
\vspace{2mm}
\caption{Error in clustering for the SSC and SSmC methods for the YaleBCrop data \cite{Elhamifar:2012uz}. The number of subjects is $5$. The errors are averaged over $5$ instances of randomly selecting the indicated number of images per subject.}
\label{tab:yale}
\end{center} 
\end{table*}

\subsection{Weizmann Face database}
We next perform simulations on the Weizmann face database -- http://www.wisdom.weizmann.ac.il/~vision/FaceBase/. Figure \ref{fig:Weizmann} shows examples of images of subjects taken under different lighting and pose conditions. For each image $\M{X}$ reduce the original size of the images by first downsampling by factor of $4$, followed by cropping keeping the indices $[x,y] \in [ 1:120, 1:80]$. In Table \ref{tab:weizmann} we show the performance of SSmC Vs SSC in clustering for several cases. Note that due to variations in pose, the data is well modeled using the submodules with the group of cyclic shifts along the horizontal direction and we see that SSmC performs much better compared to SSC in clustering performance. 

\begin{figure*}	
\centering
 	\includegraphics[height= 2in, width = 4.5 in]{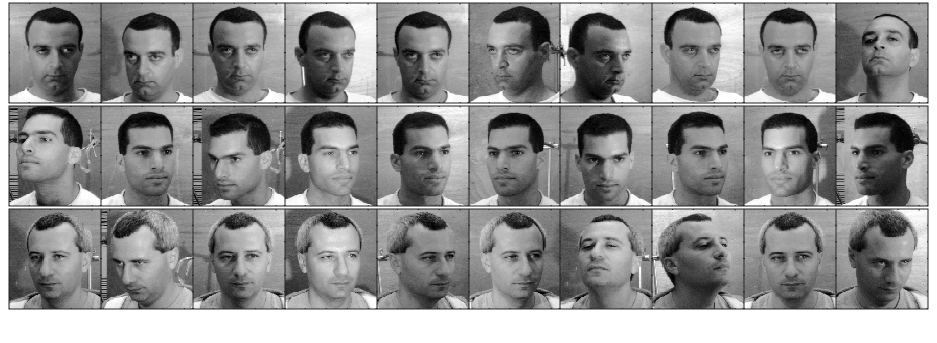}
 	\caption{\label{fig:Weizmann} Examples of faces from 3 subjects from the Weizmann data set taken under different lighting and pose conditions. In this database, the minor variations in the pose can be approximately captured by shifts along the horizontal direction.}
	\end{figure*}   

\begin{table*}
\begin{center}
\begin{tabular}{ |c|c| c|c| }
\hline
 SSC & $20$ images per sub. & $25$ images per sub. & $35$ images per sub. \\
 \hline 
 SSmC & 27.50 \% & 11 \%  & 13.29 \% \\  
 \hline
 SSC & 46.75 \% & 45.60\% & 49.57\% \\
 \hline
\end{tabular}
\vspace{2mm}
\caption{Error in clustering for the SSC and SSmC methods for the Weizmann data. The number of subjects is $4$. The errors are averaged over $5$ instances of randomly selecting (without replacement) the indicated number of images per subject.}
\label{tab:weizmann}
\end{center} 
\end{table*}

\section{Conclusions and future work}
\label{sec:conclude}

The approach presented in this paper is quite general and can be extended to rotation invariant clustering, by first taking the Radon transform and then using the shift-invariant approach described above in the Radon domain. This is because rotations corresponds to shifts along the angle axis in the Radon transform \cite{Chen:2012eo}. Finally one can also combine the two methods by considering direct product of the two or more groups. However, note that the computational complexity of the algorithm also increases with the group size and the size of the data. An important aspect of the future work will be to look at computationally efficient approaches similar to thresholded subspace clustering \cite{Heckel_ArXiv13} using the re-definition of angle between subspace elements to angle between submodule elements. 

\section{Acknowledgements}
Shuchin Aeron would like to acknowledge the support by NSF through grant CCF:1319653. 

%


\bibliographystyle{IEEEbib}
\bibliography{SSmC_Geom,SSmCbibliography}

\end{document}

%% file: macrofile.tex
\newcommand{\Sec}[1]{\hyperref[sec:#1]{\S\ref*{sec:#1}}} 
\newcommand{\Eqn}[1]{\hyperref[eq:#1]{(\ref*{eq:#1})}} 
\newcommand{\Fig}[1]{\hyperref[fig:#1]{Figure~\ref*{fig:#1}}} 
\newcommand{\Tab}[1]{\hyperref[tab:#1]{Table~\ref*{tab:#1}}} 
\newcommand{\Thm}[1]{\hyperref[thm:#1]{Theorem~\ref*{thm:#1}}} 
\newcommand{\Lem}[1]{\hyperref[lem:#1]{Lemma~\ref*{lem:#1}}} 
\newcommand{\Prop}[1]{\hyperref[prop:#1]{Property~\ref*{prop:#1}}} 
\newcommand{\Cor}[1]{\hyperref[cor:#1]{Corollary~\ref*{cor:#1}}} 
\newcommand{\Def}[1]{\hyperref[def:#1]{Definition~\ref*{def:#1}}} 
\newcommand{\Alg}[1]{\hyperref[alg:#1]{Algorithm~\ref*{alg:#1}}} 
\newcommand{\Ex}[1]{\hyperref[ex:#1]{Example~\ref*{ex:#1}}} 

\newcommand{\rs}{ {\tt reshape}}

\newcommand{\V}[1]{{\bm{\mathbf{\MakeLowercase{#1}}}}} 
\newcommand{\mb}[1]{\mathbb{#1}}
\newcommand{\mc}[1]{\mathcal{#1}}

\newcommand{\M}[1]{{\bm{\mathbf{\MakeUppercase{#1}}}}} 

\newcommand{\T}[1]{\boldsymbol{\mathscr{\MakeUppercase{#1}}}} 

\newcommand{\la}{\langle}
\newcommand{\ra}{\rangle}
